\documentclass{article}[12pt]
\usepackage[utf8]{inputenc}

\usepackage[english]{babel}
\usepackage{amsmath} 
\usepackage{natbib}
\RequirePackage[colorlinks,citecolor=blue,urlcolor=blue]{hyperref}
\usepackage{graphicx}
\usepackage{amsfonts,amssymb}
\usepackage{epsfig}%
\usepackage{setspace}
\def\sqr#1#2{{\vcenter{\vbox{\hrule height.#2pt
  \hbox {\vrule width.#2pt height#1pt \kern#1pt
    \vrule width.#2pt}
  \hrule height.#2pt}}}}
\def\square{\mathchoice\sqr56\sqr56\sqr45\sqr34}
\def\done{\rightline{$\square$}}

\newtheorem{theorem}{Theorem}[section]
\newtheorem{lemma}{Lemma}[section]
\newtheorem{proposition}[theorem]{Proposition}

\newenvironment{proof}{{\it Proof.}}{\newline \done}

\newcommand{\TT}{\mathbf{T}}

\newcommand{\VV}{\mathbf{V}}

\newcommand{\lo}{\mathrm{low}}
\newcommand{\up}{\mathrm{upp}}
\providecommand{\Pr}{\mathrm{Pr}}
\providecommand{\1}{\mathbf{1}}

\begin{document}
\title{A One-way ANOVA Test for Functional Data with Graphical Interpretation
}
\date{}
\maketitle


\centerline{Tom\'a\v s Mrkvi\v cka}
{Faculty of Economics, University of South Bohemia, Studentsk{\'a} 13, 37005 \v{C}esk\'e Bud\v{e}jovice, Czech Republic,               {email:mrkvicka.toma@gmail.com} }  \\
\centerline{Mari Myllym{\"a}ki}
{Natural Resources Institute Finland (Luke), Latokartanonkaari 9, FI-00790 Helsinki, Finland}\\

\centerline{Milan J\'ilek}
{Faculty of Economics, University of South Bohemia, Studentsk{\'a} 13, 37005 \v{C}esk\'e Bud\v{e}jovice, Czech Republic }\\
\centerline{Ute Hahn}
{ Centre for Stochastic Geometry and Advanced Bioimaging, Department of Mathematics, University of Aarhus, 8000 Aarhus C, Denmark   }


\begin{abstract}
\ A new functional ANOVA test, with a graphical interpretation of the result, is presented. 
The test is an extension of the global envelope test introduced by Myllym{\"a}ki et al.\ (2017, Global envelope tests for spatial processes, J.\ R.\ Statist.\ Soc.\ B 79, 381--404, doi: 10.1111/rssb.12172).
The graphical interpretation is realized by a global envelope which is drawn jointly for all samples of functions. If a mean function computed from the empirical data is out of the given envelope, the null hypothesis is rejected with the predetermined significance level $\alpha$.
The advantages of the proposed one-way functional ANOVA are that it identifies the domains of the functions which are responsible for the potential rejection. 
We introduce two versions of this test: the first gives a graphical interpretation of the test results in the original space of the functions and the second immediately offers a post-hoc test by identifying the significant pair-wise differences between groups. 
The proposed tests rely on discretization of the functions, therefore the tests are also applicable in the multidimensional ANOVA problem.
In the empirical part of the article, we demonstrate the use of the method by analyzing fiscal decentralization in European countries. 


\end{abstract}

\section{Introduction}
Functional data appear in a number of scientific fields, where the process of interest is monitored continuously. Those include e.g.\ monitoring of the share price, the temperature in a given location or monitoring of a body characteristic. A classical statistical problem is to decide, if there exist differences between the groups of functions (e.g.\ the control  and treatment group). This problem is usually solved by determining the differences among the mean group functions and then we deal with a one-way functional ANOVA  problem.

The functional ANOVA problem, both one-way and more complex designs, was previously studied by many authors.
For example, \cite{CuevasEtal2004} introduced an asymptotic version of the ANOVA $F$-test, and \cite{Zhang2014} considered asymptotic or bootstrapped versions of a $L_2$-norm based test, $F$-type statistic based test and globalizing pointwise $F$-test. Further, \cite{GoreckiSmaga2015} introduced a method based on a basis function representation, \cite{RamsaySilverman2006} described a bootstrap procedure based on pointwise $F$-tests, \cite{AbramovichAngelini2006} used wavelet smoothing techniques, and \cite{FerratyEtal2007} used a dimension reduction approach. Furthermore, \cite{CuestaFebrero2010} applied the $F$-test on several random univariate projections and bound the tests together through the false discovery rate (FDR).
There is also the possibility to transform the functions into single numbers and use the classical ANOVA, but such a procedure can be blind against some alternatives.

Furthermore, nonparametric permutation procedures have been used to address this problem. \cite{Hahn2012} used a one-dimensional integral deviation statistic to summarize the deviance between groups. 
Its distribution was obtained by permuting the functions. \cite{NicholsHolmes2001} {based the test} either on certain pointwise statistics, such as the $F$-statistic, and found the distribution of its maxima by permutation, or {alternatively used} the size of the area which is given by exceeding some given threshold. 
Since these statistics need to satisfy the homogeneity across the functional domain, \cite{PantazisEtal2005} recommended to concentrate on the $p$-values which are implicitely homogeneous across the domain and find the distribution of its minima by permutation. This $p$-min and also $F$-max methods are able to identify the regions of rejections by identifying a threshold of the statistics of the interest. 

To identify the regions of rejections also other methods were developed in the literature. \citet{CoxLee2008} developed  a method similar to the $p$-min procedure, but specifically designed for functional data. They applied the Westfall-Young randomization method to correct for multiple testing. A global $p$-value can not be obtained for this method. \citet{VsevolozhskayaEtAl2014} partitioned the domain of interest and applied multiple testing correction on the individual parts of the domain. A disadvantage is that partition has to be prespecified.
\citet{PiniEtAl2018} developed a functional multi-way ANOVA framework for determining the regions of rejections under the interval wise control of the error rate.  This approach controls false rejection on any interval, whereas typically the family wise error rate, i.e.\ false rejection of any pointwise hypothesis, is controlled. Also the regions with differences between groups can not be determined by this method. \citet{ChoiReimherr2018} defined a graphical representation of confidence bands for functional data inference, which can be used for two functional sample testing.

However, none of the available methods is able to give a graphical interpretation of group specific differences with respect to the null hypothesis or pairwise group differences, together with providing the global $p$-value in the family wise error rate sense. This interpretation can help the user to understand what are the reasons of potential rejections, when or where the potential differences appear. Our new proposed method which has such a graphical interpretation is based on the global rank envelope test and the extreme rank length measure introduced in \cite{MyllymakiEtal2017}. Here we extend this procedure for the functional ANOVA setting by using permutations to obtain the simulations from the null model which are required for applying the global rank envelope test. We define the test statistics suitable for testing the functional ANOVA null hypothesis. Further, we introduce here new graphical interpretation of the test based on the extreme rank length directly which allows to merge the graphical interpretation with single extreme rank length $p$-value rather than with $p$-interval as it was merged in the original definition of global rank envelope test. In this work, we concentrate on the one-way functional ANOVA problem only, because in this case and under homoscedasticity the proposed Monte Carlo test is exact, i.e.\ its type I error is exactly the prescribed significance level $\alpha$, {under the assumption of all observed functions being from the same distribution}. We also define an extension of our method for the heteroscedastic case. 

In Section \ref{sec:ANOVA} we introduce two versions of our completely nonparametric method for the one-way functional ANOVA problem. The graphical interpretation of the second version of the test also gives an immediate post-hoc test {for finding which of the groups differ from each other}. Interestingly, this post-hoc comparison is done simultaneously with the ANOVA test, thus at the exact significance level $\alpha$. Therefore, no second comparison is needed to find which groups differ.

In Section \ref{sec:fANOVA} we also describe the global rank envelope test applied to the pointwise $F$-statistics. This test does not have its graphical representation in the space of the functions, but we introduce it as another possibility of applying the global envelope test in the functional ANOVA setting. Further, in Section \ref{sec:HS} we show, how these methods can be used to testing the homoscedasticity.

In Section \ref{sec:SS}, we present results of a simulation study that was performed to compare powers of our graphical procedures with the powers of the procedures which were already available. In order to be able to compare the performances, we chose to the comparison only such procedures which were available through the software R \citep{REnvir2019} and which provided the global $p$-value in the sense of the family wise error rate. The simulation study was performed only in order to show that our method, which has a unique graphical interpretation, has comparable power with respect to other global methods. All our proposed methods can be found in the R package GET {\citep{MyllymakiEtal2017, MyllymakiMrkvicka2019}}.

In Section \ref{sec:DS}, we apply our methods to the fiscal decentralization issue in European countries. The empirical analysis aims to capture differences in developments 
of government expenditure decentralization among different groups of European countries. The assumption is, based on the existing literature, that countries with a longer European integration history and therefore presumably with deeper economic and political integration are supposed to decentralize their government expenditure more extensively. We use the government expenditure centralization ratios of 29 European Union and EFTA countries in period from 1995 to 2016 sorted into three groups according to the presumed level of European economic and political integration.

Section \ref{sec:discussion} is for further discussion.

\section{Graphical functional ANOVA} \label{sec:ANOVA}

Let us assume that we have $J$ groups which contain $n_1, \ldots , n_J$ functions {observed on the finite interval $R=[a, b]$} and denote the functions by $T_{ij}, j=1, \ldots, J, i=1, \ldots , n_j$.
Assume that $\{T_{ij}\in L^\infty, i=1, \ldots , n_i\}$ is an i.i.d. sample from a stochastic process $SP(\mu_j, \gamma_j)$ with a mean function $\mu_j$ and a covariance function $\gamma_j(s, t), s, t \in R$ for $j=1, \ldots , J$.
We want to test the hypothesis
$$H_0 : \mu_1(r) = \ldots = \mu_J(r), r \in R.$$
We do not need to specify the stochastic process  $SP$ in order to define our method and thus our method can be taken as completely nonparametric comparison of groups of functions. 

The hypothesis $H_0$ is equivalent to the hypothesis
$$
H'_0 : \mu_{j'}(r) - \mu_j(r) = 0, r \in R,  j' =1, \ldots , J-1, j =j', \ldots , J.
$$
This hypothesis corresponds to the post-hoc test done usually after the ANOVA test is significant.

In the following we introduce the test statistics both for the hypothesis $H_0$ and $H'_0$, first for the case of equal covariance functions (i.e.\ for the case of $\gamma_1(s,t) = \ldots = \gamma_J(s,t), s, t \in R$) and then for the case of unequal variance functions (i.e.\ for the case of $\gamma_1(s,t) / \gamma_1(s,s)  = \ldots = \gamma_J(s,t) / \gamma_J(s,s), s, t \in R$) (Sections \ref{sec:one-way-group-comparison} and \ref{sec:fANOVA}). Then we describe how the permutations are performed under the null hypothesis (Section \ref{sec:permutations}) and show how the rank envelope test can be used for these test statistics (Section \ref{sec:GET}). 

The implementation of our method relies on the discretization of functions. We assume that all functions are discretized in the same way obtaining values at points $(r_1, \ldots , r_K)$. If this is not the case, we have to apply smoothing techniques \citep[e.g.\ those described in][]{Zhang2014} and then make such a necessary discretization. In the simulation study (see Section \ref{sec:SS}) we study our method with respect to increasing denseness of the discretization. Remark here, that the discretization can be arbitrary and equidistances are not required.

\subsection{Test vectors}\label{sec:one-way-group-comparison}

The hypothesis $H_0$ can be tested by the test vector consisting of the mean of functions in the first group followed by the mean of test functions in the second group, etc.
{Shortly, the test vector is}
\begin{equation}\label{TT}
\TT = (\overline{T}_1({\bf r}), \overline{T}_2({\bf r}), \ldots , \overline{T}_J({\bf r})),
\end{equation}
where $\overline{T}_j({\bf r}) = (\overline{T}_j(r_1), \ldots , \overline{T}_j(r_K))$.
Thus, the length of the test vector becomes $J \times K$, where $K$ stands for the number of $r$ values to which the functions are discretized.

The hypothesis $H'_0$ can be tested by the test vector consisting of the differences of the group means of functions{, i.e. the test vector is}
\begin{equation}\label{TTprime}
\TT' = (\overline{T}_1({\bf r}) - \overline{T}_2({\bf r}), \overline{T}_1({\bf r}) - \overline{T}_3({\bf r}), \ldots , \overline{T}_{J-1}({\bf r}) - \overline{T}_J({\bf r})).
\end{equation}
Here the length of the test vector becomes $J (J-1) / 2 \times K$.

\subsubsection{Correction for an unequal variances} \label{corrH0}

To deal with different variances of functions in different groups, 
we consider the rescaled functions $S_{ij}$ instead of the original functions $T_{ij}$,

\begin{equation}\label{Sij1}
S_{ij}(r) = \frac{T_{ij}(r) - \overline{T_j}(r)}{\sqrt{\text{Var}(T_j(r))}}\cdot \sqrt{\text{Var}(T(r))} + \overline{T_j}(r),
\end{equation}
where the group sample mean $\overline{T_j}(r)$ and overall sample variance $\text{Var}(T(r))$ are involved
to keep the mean and variability of the functions at the original scale. The group sample variance $\text{Var}(T_j(r))$ corrects the unequal variances. 

For small samples, the sample variance estimators 
can have big variance, which may influence the test procedure undesirably.
In order to deal with this problem the variances can be smoothed by applying moving average (MA) to the estimated variance with a chosen window size $b$.
Thus, the rescaled functions take the form
\begin{equation}\label{Sij2}
S_{ij}(r) = \frac{T_{ij}(r) - \overline{T_j}(r)}{\sqrt{\text{MA}_b(\text{Var}(T_j(r)))}}\cdot \sqrt{\text{MA}_b(\text{Var}(T(r)))} + \overline{T_j}(r).
\end{equation}

After transformation, the test vectors are composed in the same way as in the case of the equal covariance functions but with rescaled functions:
\begin{equation}\label{TT_1}
\TT_s = (\overline{S}_1({\bf r}), \overline{S}_2({\bf r}), \ldots , \overline{S}_J({\bf r})).
\end{equation}
where $\overline{S}_j({\bf r}) = (\overline{S}_j(r_1), \ldots , \overline{S}_j(r_K)$, and
\begin{equation}\label{TT_1prime}
\TT'_s = (\overline{S}_1({\bf r}) - \overline{S}_2({\bf r}), \overline{S}_1({\bf r}) - \overline{S}_3({\bf r}), \ldots , \overline{S}_{J-1}({\bf r}) - \overline{S}_J({\bf r})).
\end{equation}

\subsection{Rank envelope $F$-type test}\label{sec:fANOVA}
When a graphical interpretation for group specific differences is not of interest but the area of rejection is, one can utilize the $F$-type test for each $r\in R$ separately and form the test vector from the $r$-wise $F$-statistics,
$$\TT_F = (F(r_1), F(r_2), \ldots , F(r_K)),$$ 
where $F(r_k)$ stands for the $F$-statistic computed from the univariate ANOVA model at point $r_k$.
In this case the correction for unequal variances can be done by choosing the variance corrected $F$-statistic.

\subsection{Homoscedasticity tests} \label{sec:HS}
Following the Levene's test of homoscedasticity, it is possible to test the equality of variances in functional ANOVA design by setting the test vector 
$\TT_V = \TT$, $\TT_V = \TT'$ or $\TT_V=\TT_F$, where instead of the original functions $T_{ij}$, the functions
\begin{equation}\label{eq:variance}
Z_{ij}(r)=|T_{ij}(r)-\overline{T}_j(r)|
\end{equation}
are used.
Similarly, in order to test the equality of lag $s$ covariance, it is possible to set the test vector $\TT_C = \TT$, $\TT_C = \TT'$ or $\TT_C=\TT_F$, where instead of the original functions $T_{ij}$, the functions
\begin{equation}\label{eq:covariance}
W_{ij}(r)=\sqrt{|(T_{ij}(r)-\overline{T}_j(r))(T_{ij}(r+s)-\overline{T}_j(r+s))|}\cdot\text{sign}[(T_{ij}(r)-\overline{T}_j(r))(T_{ij}(r+s)-\overline{T}_j(r+s))],    
\end{equation}
for $r\in [a, b-s]$, are used.

\subsection{Permutations and exchangebility of the test vectors}\label{sec:permutations}
The most important aspect of the permutation tests is the manner in which data are shuffled under the null hypothesis. In all our one-way ANOVA tests, we perform the simple permutation of raw functions among the groups. That is, if $G$ is a vector of group indices of length $N=\sum_{j=1}^J n_j$, then the permutation $\mathbf P$ is $N \times N$ matrix that has all elements being either 0 or 1, each row and column having exactly one 1.  Pre-multiplicating the group indices $G$ by the matrix $\mathbf P$ permutes the group indices. 
Note that the possible correction for unequal variances is performed prior to the permutation and the permutations are consequently done for the rescaled functions \eqref{Sij1} or \eqref{Sij2}. 

We say that a test vector $\TT$ is exchangeable if the observed and simulated (permuted) test vectors $\TT^1, \dots, \TT^s$ are exchangeable, i.e.\ the joint distribution of $\TT^1, \dots, \TT^s$ is not affected by permutation.

\begin{proposition}\label{proposition1}
Under the assumption that all the functions $T_{ij}, j=1, \ldots , J, i=1, \ldots , n_j$ follow the same stochastic process, the test vectors $\TT$, $\TT'$, $\TT_F$ and also $\TT_V$, $\TT_C$ are exchangeable for permutations $\mathbf P$. Under the assumption of normality of stochastic processes $SP(\mu_j,\gamma_j)$ the test vectors $\TT_s$, $\TT'_s$ and $\TT_F$ with variance corrected $F$-statistics are asymptotically exchangeable for permutations  $\mathbf P$ for the case of unequal variances and the null hypothesis of equal means. The asymptotics is taken over $\min_{j=1}^J n_j$.
\end{proposition}
\begin{proof}
Since the permutations are performed on the whole functions (i.e.\ the block permutation scheme is used) and we assume that the functions form an i.i.d.\ sample from a stochastic process, the joint distribution of $\TT$ ($\TT'$, $\TT_F$, $\TT_V$, $\TT_C$) is equal to the joint distribution of $\TT$ ($\TT'$, $\TT_F$, $\TT_V$, $\TT_C$) for permuted  groups $\mathbf P G$.

In the case of unequal variances the functions are first scaled by the sample group variance for computation of $\TT_s$ and $\TT'_s$. The sample group variance  $\text{Var}(\overline{T}_j(r))$ converges a.s.\ to the true group variance. {This holds similarly} for the group sample mean and overall sample variance. Thus the stochastic process $S_{ij}$ converges in distribution to $SP(\mu,\gamma_j(s,t)\gamma(s,s)/\gamma_j(s,s))$, where $\gamma(s,s)$ is the overall variance. Under the null hypotheses of equal means and unequal variances and assumption of normality these stochastic processes are equal and thus the test vectors are asymptotically exchangeable. A similar proof can be made for $\TT_F$ in the case of unequal variances. 
\end{proof}

\subsection{ Global rank envelope test}\label{sec:GET}
The idea of the global rank envelope was introduced in \cite{MyllymakiEtal2017} for {testing in spatial statistics}. Further \citet{MrkvickaEtal2017} extended the notion of this global envelope for general multivariate test vectors. This extension applies, e.g., to the case where the multivariate vector consists of values of two or more functions at once. We first recall the measures and associated $p$-values introduced in \cite{MyllymakiEtal2017}. Second, we define the global extreme rank length  envelope as a refinement of the global rank envelope.

Assume the general multivariate vector of the form
\[
\VV = \big(V_1,\dots,V_d).
\]
Let $\VV_1, \dots, \VV_{s}$ be the multivariate vectors generated by permutations under the null hypothesis. Let $\VV_1$ denote the vector obtained by identical permutation.

First we define the extreme rank $R_i$ of the vector $\VV_i=(V_{i1},\dots,V_{id})$ as the minimum of its pointwise ranks, namely
  \begin{equation}\label{eq:minglobalrank}
  R_i=\min_{k=1,\ldots,d} R_{ik},
  \end{equation}
where the pointwise rank $R_{ik}$ is the rank of the element $V_{ik}$ among the corresponding elements $V_{1k}, V_{2k}, \dots, V_{s k}$ of the $s$ vectors such that the lowest ranks correspond to the most extreme values of the statistics.
How the pointwise ranks are determined, depends on whether a one-sided or a two-sided envelope test is to be performed: Let $r_{1k}, r_{2k}, \dots, r_{s k}$ be the raw ranks of $V_{1k}, V_{2k}, \dots, V_{s k}$, such that the smallest $V_{ik}$ has rank 1. In the case of ties, the raw ranks are averaged. The pointwise ranks are then calculated as
\begin{equation}
  R_{ik}=\begin{cases}
    r_{ik}, &\text{for one-sided test, small $V$ is considered extreme}\\
    s+1-r_{ik}, &\text{for one-sided test, large $V$ is considered extreme}\\
   \min(r_{ik}, s+1-r_{ik}), &\text{for two-sided test}.
  \end{cases}
\end{equation}

The extreme ranks can contain many ties, e.g. in a one-sided test with $d$-variate vectors, up to $d$ out of the $s$ vectors can take the rank 1. Therefore we need to break these ties in an efficient way. Ordering of the vectors by the extreme rank length \citep{MyllymakiEtal2017} refines the extreme rank ordering in order to minimize the possibility of ties. 

Consider the vectors of pointwise ordered ranks $\mathbf{R}_i=(R_{i[1]}, R_{i[2]}, \dots , R_{i[d]})$, where $\{R_{i[1]},  \dots , R_{i[d]}\}=\{R_{i1}, \dots, R_{id}\}$ and $R_{i[k]} \leq R_{i[k^\prime]}$ whenever $k \leq k^\prime$.
The extreme rank given in \eqref{eq:minglobalrank} corresponds to $R_i=R_{i[1]}$. The extreme rank length measure of the vectors $\mathbf{R}_i$ is equal to
\begin{equation}
\label{eq:lexicalrank}
   R_i^{\text{erl}} = \frac{1}{s}\sum_{i^\prime=1}^{s} \1(\mathbf{R}_{i^\prime} \prec \mathbf{R}_i),
\end{equation}
where
\[
\mathbf{R}_{i^\prime} \prec \mathbf{R}_{i} \quad \Longleftrightarrow\quad
  \exists\, n\leq d: R_{i^\prime[k]} = R_{i[k]} \forall k < n,\  R_{i^\prime[n]} < R_{i[n]}.
\]

We remark here that \citet{NarisettyNair2016} independently defined the two-sided extreme rank length measure as a functional depth and called it extremal depth. 

\subsubsection{$p$-values}
We distinguish three different $p$-values attached to the rank envelope test. All the $p$-values are based on Monte Carlo testing principles.
The conservative and liberal $p$-values are given as
  \begin{equation}\label{eq:p-value_globalranktest}
   p_+ =  \sum_{i=1}^{s} \1 (R_i \leq R_1)  \big/ s, \quad  p_- =  \sum_{i=1}^{s} \1 (R_i < R_1)  \big/ s.
  \end{equation}
The $p$-value based on the extreme rank length ordering is given as
 \begin{equation}\label{eq:p-value_globalranktest_erl}
   p^{\text{erl}} =  1 - \sum_{i=1}^{s} \1 (\mathbf{R}_1 \prec \mathbf{R}_i)  \big/ s.
  \end{equation}
According to \citet[Proposition 6.1]{MyllymakiEtal2017}, it holds that $p_- < p^{\text{erl}} \leq p_+$.

Note here, that there still can appear some ties between $\mathbf{R}_1$ and $\mathbf{R}_i, i=2, \ldots , s$. However, since these ties are unlikely to happen, we define $p^{\text{erl}}$ as the conservative $p$-value. 
Alternatively these ties could be broken by randomization.

\subsubsection{The new graphical envelope}
\citet{MyllymakiEtal2017} defined the graphical global envelope with respect to the ordering of the extreme ranks $R_i, i=1, \ldots , s$. This ordering can have a lot of ties and consequently the graphical envelope based on this ordering requires a lot of permutations in order to be precise. We eliminate this problem in this paper by defining the graphical envelope with respect to the extreme rank length ordering \eqref{eq:lexicalrank}. 

Assuming that all the $\VV_i$ follow the same joint distribution, we construct rank envelopes with level $(1-\alpha)$ as sets $\{\VV_{\lo}^{(\alpha)}, \VV_{\up}^{(\alpha)}\}$ such that, the probability that $\VV=(V_{1}, \dots, V_{d})$ falls outside this envelope in any of the $d$ points is less or equal to $\alpha$,
\begin{equation*}
  \Pr\big(\exists\, k\in \{1,\dots,d\}: V_{k}\notin[V_{\lo\,k}^{(\alpha)},V_{\up\,k}^{(\alpha)}])\leq \alpha.
\end{equation*}
{
Let $R_{(\alpha)}^{\text{erl}}$ be the largest value in $\{R_1^{\text{erl}},\dots,R_{s}^{\text{erl}}\}$ for which
\begin{equation}\label{eq:Ralpha}
  \sum_{i=1}^{s} \1 \left(R_i^{\text{erl}} < R_{(\alpha)}^{\text{erl}}\right) \leq \alpha s,
\end{equation}
and $I_\alpha = \{i\in 1,\dots, s: R_i^{\text{erl}} \geq R_{(\alpha)}^{\text{erl}}  \}$ be the index set of vectors whose extreme rank length measure is larger than or equal to the critical value $R_{(\alpha)}^{\text{erl}}$.
Then}
define
\begin{equation*}
  V_{\lo\,k}^{(\alpha)}= \min_{i\in I_\alpha}V_{ik}, \quad   V_{\up\,k}^{(\alpha)}= \max_{i\in I_\alpha}V_{ik}
\end{equation*}
for the two-sided test, following the idea of \cite{NarisettyNair2016}. For one-sided tests, let $V_{\lo\,k}^{(\alpha)}=-\infty$ or $V_{\up\, k}^{(\alpha)}=\infty$, respectively, for all $k=1,\ldots,d$.
{This envelope has} the graphical interpretation described in the next subsection.

The following theorem states that inference based on the $p^{\text{erl}}$ and the global envelope specified by $ V_{\lo\,k}^{(\alpha)}$ and $V_{\up\,k}^{(\alpha)}$ are equivalent. Therefore, we can refer to this envelope as the $100\cdot(1-\alpha)$\% global extreme rank length envelope.
\begin{theorem}\label{thm:envelope-vs-pinterval}
Let $p^{\text{erl}}$ be as given in \eqref{eq:p-value_globalranktest_erl}, and $ V_{\lo\,k}^{(\alpha)}, V_{\up\,k}^{(\alpha)}$ define the $100\cdot(1-\alpha)$\% global extreme rank length envelope. Then, assuming that there are no pointwise ties with probability $1$, it holds that:
\begin{enumerate}
 \item $V_{1k} < V_{\lo\,k}^{(\alpha)}$ or $V_{1k} > V_{\up\,k}^{(\alpha)}$ for some $k = 1, \ldots , d$
iff $p^{\text{erl}} \leq \alpha$, in which case the null hypothesis is rejected;
 \item $V_{\lo\,k}^{(\alpha)} \leq V_{1k} \leq V_{\up\,k}^{(\alpha)}$ for all $k = 1, \ldots , d$ iff $p^{\text{erl}} > \alpha$, and thus the null hypothesis is not rejected;
\end{enumerate}
\end{theorem}
\begin{proof}
According to the definition of $p^{\text{erl}}$ is $p^{\text{erl}} \leq \alpha$ iff number of  $\mathbf{R}_i$ smaller or equal to $\mathbf{R}_1$ is smaller or equal to $\alpha s$. That is equivalent, according to the definition of $R_{(\alpha)}^{\text{erl}}$, to the $R_1^{\text{erl}} < R_{(\alpha)}^{\text{erl}}$. This holds iff $1 \notin I_\alpha$, which is equivalent to $V_{1k} < V_{\lo\,k}^{(\alpha)}$ or $V_{1k} > V_{\up\,k}^{(\alpha)}$ for some $k = 1, \ldots , d$ according to the definition of the extreme rank length envelope.
The second part of the proof can be proven equivalently.
\end{proof}

In the following theorem, we will prove that the global extreme rank length envelope is contained in the global rank envelope.
The $l$-th rank envelope was defined in \cite{MyllymakiEtal2017} by
\begin{equation}\label{kth_envelopes}
V^{(l)}_{\lo \, k}= \underset{i=1,\dots,s}{{\min}^{l}} V_{i k}
\quad\text{and}\quad
V^{(l)}_{\up \, k}= \underset{i=1,\dots,s}{{\max}^{l}} V_{i k} \quad \text{for } k = 1, \ldots , d,
\end{equation}
where $\min^l$ and $\max^l$ denote the $l$-th smallest and
largest values, respectively, and $l=1,2,\dots,\lfloor (s+1)/2\rfloor$.

\begin{lemma}\label{lemmaWl} Let
\begin{equation*}
  W_{\lo\,k}^{(l)}= \min_{i\in I_l}V_{ik}, \quad   W_{\up\,k}^{(l)}= \max_{i\in I_l}V_{ik}
\end{equation*}
for $I_l = \{i\in 1,\dots, s: R_i \geq l  \}$.
Then $W_{\lo\,k}^{(l)}\geq V_{\lo\,k}^{(l)}$ and $W_{\up\,k}^{(l)}\leq V_{\up\,k}^{(l)}$.
\end{lemma}

\begin{proof}
Since all vectors in $I_l$ have the extreme rank greater or equal to $l$ and thus $R_{ik} \geq l$ or $R_{ik} \leq s-l+1$ (for two-sided test), the envelope defined by $W$ is contained in the envelope defined by $V$.
\end{proof}

\begin{theorem}
The $100\cdot(1-\alpha)\%$ global extreme rank length envelope is contained in the $100\cdot(1-\alpha)\%$ global rank envelope.
\end{theorem}
\begin{proof}
This is a direct consequence of Lemma \ref{lemmaWl} and the fact that
$I_\alpha$ contains a smaller number of functions than $I_l$ for $l=l_\alpha = \max\left\{l:\  \sum_{i=1}^{s}\1(R_i<l) \leq \alpha s\right\}$, which is the critical rank for the  $100\cdot(1-\alpha)\%$ global rank envelope.
\end{proof}

Remark here that the $(l+1)$-th rank envelope given by $V^{(l+1)}_{\lo \, k}$ and $V^{(l+1)}_{\up \, k}$ is not necessarily contained in the $l$-th extreme rank length envelope given by $W_{\lo\,k}^{(l)}$ and $W_{\up\,k}^{(l)}$.

\subsection{One-way graphical functional ANOVA test}
The proposed tests are performed in three steps. First the test vector is chosen. Second $s$ permutations are applied to the raw functions (or on the rescaled functions in the case of unequal variance) and the chosen test vector is computed for each permutation. Third the {global} rank envelope test is applied to the set of $s$ test vectors. The following theorem specifies the graphical interpretation of our proposed tests and claims the exactness of the graphical method. 

\begin{theorem}\label{thm:envelope-vs-pinterval-T}
Consider one-way graphical functional ANOVA test with $\TT$, $\TT'$ or $\TT_F$ chosen as the test vector. Assume that all the functions $T_{ij}, j=1, \ldots , J, i=1, \ldots , n_j$ follow the stochastic process $SP(\mu_j,\gamma)$.
Let $p^{\text{erl}}$ be as given in \eqref{eq:p-value_globalranktest_erl}, and $T^{{\alpha}}_{\lo \, k}$ and $T^{{\alpha}}_{\up \, k}$ define the $100\cdot(1-\alpha)$\% global extreme rank length envelope. Then, assuming that there are no pointwise ties with probability $1$ in the stochastic process $SP(\mu,\gamma)$, it holds that:
\begin{enumerate}
 \item $T_{1 k} < T^{{\alpha}}_{\lo \, k}$ or $T_{1 k} > T^{{\alpha}}_{\up\, k}$ for some $k$
iff $p^{\text{erl}} \leq \alpha$, in which case the null hypothesis is rejected;
 \item $T^{{\alpha}}_{\lo \, k} \leq T_{1 k} \leq T^{{\alpha}}_{\up\, k}$ for all $k$ iff $p^{\text{erl}} > \alpha$, and thus the null hypothesis is not rejected.
\end{enumerate}
\end{theorem}

Theorem \ref{thm:envelope-vs-pinterval-T} is direct consequence of Proposition \ref{proposition1} and Theorem \ref{thm:envelope-vs-pinterval}. The same theorem holds for the proposed homoscedasticity tests.
For the case of unequal variances the above interpretation is due to the Proposition \ref{proposition1} achieved only asymptotically with the additional assumption of normality.

Since we apply here the global extreme rank length envelope and not the global rank envelope as it was the case in our previous works, a lower number of permutations can be used. 
Anyway we recommend to use some thousands of permutations at minimum {for repeatability}. 
In case of many groups the number of permutations has to increased in order to not loose the power of the test, as it is demonstrated in the simulation study.

It is important to mention here that the graphical interpretion automatically identifies which groups are responsible for the potential rejection and also it identifies which parts of the functions are responsible for the rejection. This is very important for the interpretation of the result of the test.

Note that for the test vector $\TT_F$ the one-sided rank test has to be used, whereas for the other test vectors the two-sided rank test is used.

\subsection{Comparison to other permutation methods} 
The nonparametric permutation methods often used in the brain image statistics are similar to our proposed methods, therefore we would like to stress the differences. The single threshold test \citep{NicholsHolmes2001} of a certain statistic whose maximum is permuted is limited to the statistics that are homogeneous across the functional domain, in order to be sensitive in the whole functional domain and not only in the part of the domain where the functions are the most varying. The $p$-min permutation procedure used e.g.\ in \cite{PantazisEtal2005} solves this problem. 
This method can be viewed as our rank envelope $F$-type test. However, the $p$-min permutation procedure uses the conservative $p$-value of our rank envelope $F$-type test, i.e. the upper bound of the $p$-interval, $p_{+}$ in \eqref{eq:p-value_globalranktest}. On the other hand, our rank envelope $F$-type test is equipped also with the extreme rank length $p$-value which solves the problem of ties in the $p$-min distribution and therefore it significantly reduces the conservativeness of the test. 

Further, our graphical functional ANOVA test gives the graphical interpretation in the original space of functions and for each group of functions, whereas the $p$-min test gives it only in the transformed space of $p$-values and for all groups simultaneously. Therefore the $p$-min test is able only to identify the regions of rejection. Our graphical functional ANOVA test is also equipped with the global extreme rank length envelope which informs the user about the variability of the curves in the study. Finally, the graphical functional ANOVA test is defined here also for combining several post-hoc tests together in one test and therefore it indicates which two groups are different and where they are different.

\section{Simulation study}\label{sec:SS}
Our simulation study has four parts. First we compared our methods with some existing methods on a design taken from the study of \cite{CuevasEtal2004} in order to check if our methods are comparable in power and significance level to the existing methods. We chose methods which were available in the software R, especially in the packages fda.usc \citep{fda.usc} and fdANOVA \citep{fdANOVA}, and which are fundamentally different of each other. Second we checked the robustness of the studied methods with respect to heteroscedasticity. Third we changed the design from comparing three groups into comparing ten groups in order to check how much of the power is lost by having a long test vector $\TT$ or $\TT'$ with respect to other procedures. Fourth we studied the dependence of the powers on the level of discretization of the functions. 
The tests included in our study are listed in Table \ref{table:tests}.

\begin{table}[ht]
\caption{List of tests included in the simulation study with their abbreviations (Abbr.) and short description.}\label{table:tests}
\centering
\begin{tabular}{lp{4.02cm}p{9.0cm}}
  \hline
 Abbr. & Introduced/described in & Test description \\ \hline
  AsF   & \cite{CuevasEtal2004} & a bootstrapped version of the asymptotic $F$-test  \\ 
  RPM   & \cite{CuestaFebrero2010} & a random univariate projection method  \\ 
  Fb   & \cite{Zhang2014} & a bootstrapped version of the $F$-type statistic test   \\ 
  GPF   & \cite{Zhang2014} & a globalizing pointwise $F$-test  \\ 
  FP   & \cite{GoreckiSmaga2015} & a method based on a basis function representation  \\ 
  IPT   & \cite{Hahn2012} & a one-dimensional integral permutation test  \\ 
  F-max   & \cite{NicholsHolmes2001} & the $F$-max permutation procedure  \\ 
  p-min   & \cite{PantazisEtal2005} & the $p$-min permutation procedure \\ \hline
  GFAM   & here & the graphical functional ANOVA based on the test vector \eqref{TT} (group means) \\ 
  GFAC   & here & the graphical functional ANOVA based on the test vector \eqref{TTprime} (group mean contrasts) \\ 
  REF   & here &  the rank envelope $F$-type test \\  
   \hline
\end{tabular}
\end{table}


First we compared all the tests of Table \ref{table:tests} using an artificial example of $J=3$ groups and $n=10$ functions in each group observed in the interval $[0,1]$ through 100 evenly spread discrete points. Four different models with two different autocorrelation error structures were considered. {The models were}
\begin{itemize}
\item M1: $T_{ij}(r) = r (1-r) + e_{ij}(r), \; i=1, 2, 3, j=1, \ldots , 10$,
\item M2: $T_{ij}(r) = r^i (1-r)^{6-i} + e_{ij}(r), \; i=1, 2, 3, j=1, \ldots , 10$,
\item M3: $T_{ij}(r) = r^{i/5} (1-r)^{6-i/5} + e_{ij}(r), \; i=1, 2, 3, j=1, \ldots , 10$,
\item M4: $T_{ij}(r) = 1+i/50 + e_{ij}(r), \; i=1, 2, 3, j=1, \ldots , 10$.
\end{itemize}
The mean function for each model and group is shown in Figure \ref{fig:fce}.

\begin{figure}
    \centering
    \includegraphics[width=\textwidth]{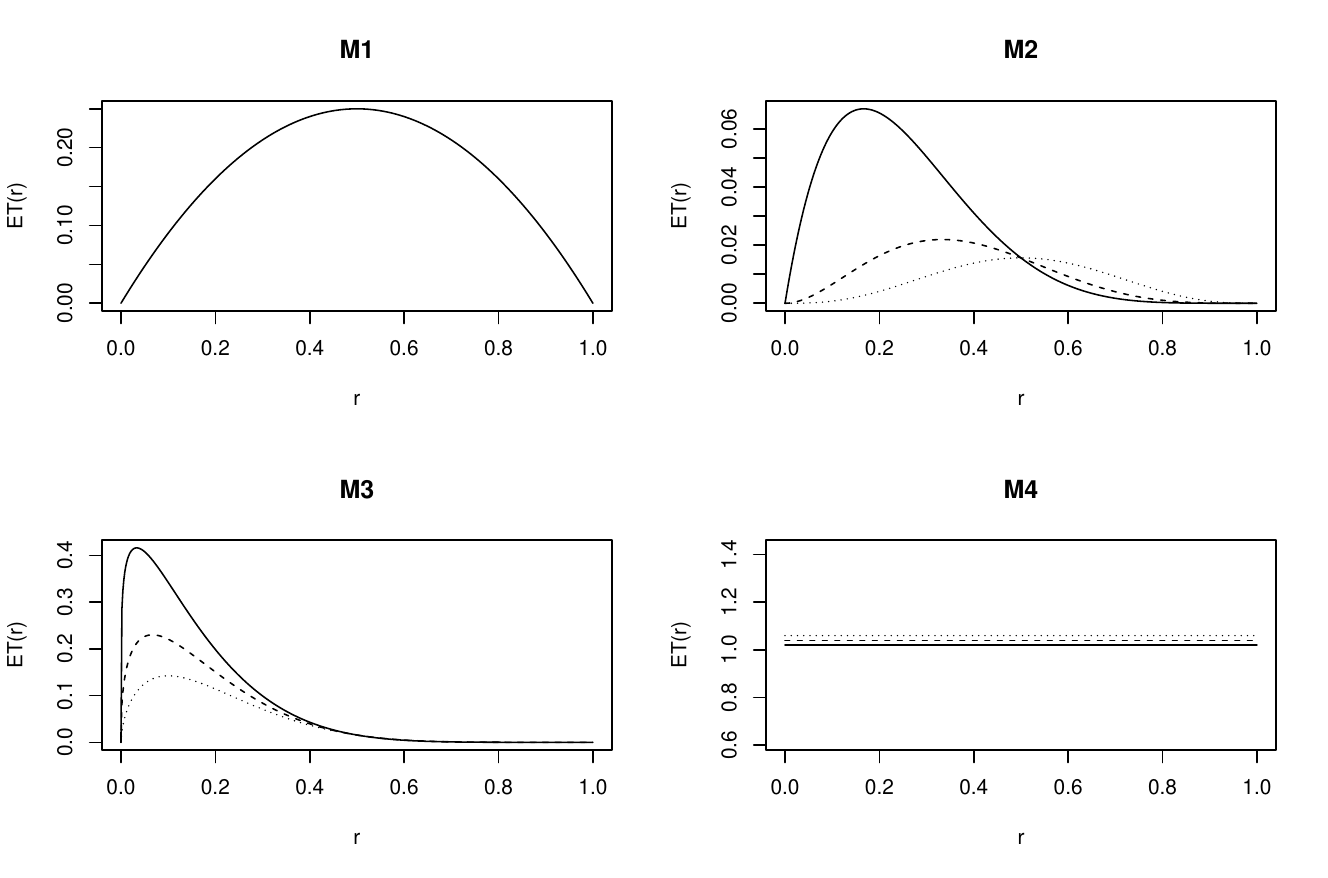}
    \caption{The mean functions of the models M1, M2, M3 and M4. The first group ($i=1$) corresponds to the solid line, second ($i=2$) dashed and third ($i=3$) dotted. {In M1, the three lines coincide.}}
    \label{fig:fce}
\end{figure}
The first autocorrelation structure of errors was modelled by the Gaussian random process with exponential correlation structure with scale parameter equal to 0.1 and standard deviation $\sigma$.  In the second structure, the errors $e_{ij}(r)$ were modelled by the  Brownian process with dispersion parameter $\sigma$. 
For each combination of the four models and two autocorrelation structures, we considered six different contaminations for the deterministic part of the model given by six different standard deviations $\sigma_1 = 0.05, \sigma_2 = 0.1, \sigma_3 = 0.15, \sigma_4 = 0.2, \sigma_5 = 0.4, \sigma_6 = 0.8$.
Every standard deviation was twice the previous one, except $\sigma_3$ which was added in the middle of $\sigma_2$ and $\sigma_4$ in order to increase the sensitivity of the simulation study.
	
The model M1 corresponds to the situation where $H_0$ is true. Thus, in this case, the predetermined significance level was estimated. It was set to $\alpha=0.05$ in all cases. The other models represent different situations where $H_0$ is false. The mean functions in the models M2 and M3 have different shape, whereas the mean functions in the model M4 are constant.

Although the permutation tests are exact, they contain variability in sampled permutations and further in the resulted $p$-value. \cite{LoosmoreFord2006} recommended to use at least 1000 Monte Carlo replicates in order to control this variability. Therefore, all the tested procedures were run using 2000 Monte Carlo replications or permutations in order to keep the running time manageable and fulfil the above requirement. The RPM was run with 30 random projections, as it was found by \cite{CuestaFebrero2010} to be high enough in a general case, and the false discovery rate $p$-value computed out of these projections was used as a final output of this procedure. The extreme rank length $p$-value was used as the output of all our new tests (GFAM, GFAC, REF). 
We performed 1000 simulations for each combination of model, autocorrelation structure and standard deviation, and 
we computed the proportion of rejections to obtain the estimates of significance levels and powers. All the results are summarized in Table \ref{iid} for the Gaussian process cases and in Table \ref{Brown} for the Brownian error cases.

The empirical significance level should be in
the interval $(0.037, 0.064)$ with the probability 0.95 (given by the 2.5\% and 97.5\% quantiles of the binomial distribution with parameters 1000 and 0.05). This was satisfied, except for the AsF and Fb procedures in the case of the Gaussian process errors. This exception was caused by the high number of discretized points. For 20 discretized points these methods did not show this feature.

Our three tests had very good power in the Brownian error case and slightly lower power than the best methods (GPF, FP, IPT) in the Gaussian process error case. {This different behaviour was obviously caused by the fact that the Gaussian process error is better captured by integral based methods, whereas Brownian error case with increasing variability is better captured by maximum based methods.} Surprisingly the REF test did not have greater powers than the RECM and RECMD tests which are fully nonparametric (the errors were normally distributed in the considered models).


\begin{table}[ht]
\caption{The proportions of rejections (at level 0.05) over 1000 runs in the case of Gaussian process errors for models M1, M2, M3, M4. See text for the model specifications and {Table \ref{table:tests} for descriptions of different test abbreviations}.}\label{iid}
\centering
\begin{tabular}{lcccccc}
  \hline
 \bf{M1} & $\sigma_1=0.05$ & $\sigma_2=0.1$ & $\sigma_3=0.15$ & $\sigma_ 4=0.2$ & $\sigma_5=0.4$ & $\sigma_6=0.8$ \\
  \hline
AsF   & 0.031 & 0.029 & 0.023 & 0.031 & 0.026 & 0.035 \\ 
  RPM   & 0.057 & 0.061 & 0.055 & 0.057 & 0.077 & 0.056 \\ 
  Fb   & 0.011 & 0.005 & 0.009 & 0.009 & 0.007 & 0.007 \\ 
  GPF   & 0.061 & 0.063 & 0.048 & 0.058 & 0.067 & 0.062 \\ 
  FP   & 0.059 & 0.055 & 0.040 & 0.056 & 0.060 & 0.050 \\ 
  IPT   & 0.062 & 0.061 & 0.042 & 0.052 & 0.060 & 0.047 \\ 
  F-max   & 0.058 & 0.063 & 0.046 & 0.051 & 0.052 & 0.048 \\ 
  p-min   & 0.034 & 0.041 & 0.025 & 0.028 & 0.035 & 0.029 \\ 
  GFAM   & 0.062 & 0.058 & 0.041 & 0.046 & 0.060 & 0.054 \\ 
  GFAC   & 0.047 & 0.061 & 0.038 & 0.043 & 0.058 & 0.051 \\ 
  REF   & 0.054 & 0.062 & 0.046 & 0.046 & 0.055 & 0.048 \\  
   \hline
 \hline
 \bf{M2} & $\sigma_1=0.05$ & $\sigma_2=0.1$ & $\sigma_3=0.15$ & $\sigma_ 4=0.2$ & $\sigma_5=0.4$ & $\sigma_6=0.8$ \\
  \hline
AsF   & 0.615 & 0.133 & 0.058 & 0.048 & 0.026 & 0.024 \\ 
  RPM   & 0.515 & 0.143 & 0.088 & 0.071 & 0.069 & 0.062 \\ 
  Fb   & 0.437 & 0.055 & 0.022 & 0.017 & 0.007 & 0.005 \\ 
  GPF   & 0.729 & 0.209 & 0.113 & 0.086 & 0.055 & 0.056 \\ 
  FP   & 0.717 & 0.197 & 0.104 & 0.084 & 0.050 & 0.054 \\ 
  IPT   & 0.699 & 0.192 & 0.104 & 0.074 & 0.051 & 0.054 \\ 
  F-max   & 0.576 & 0.128 & 0.085 & 0.067 & 0.051 & 0.050 \\ 
  p-min   & 0.486 & 0.096 & 0.053 & 0.043 & 0.030 & 0.028 \\ 
  GFAM   & 0.613 & 0.166 & 0.094 & 0.065 & 0.063 & 0.055 \\ 
  GFAC   & 0.600 & 0.144 & 0.084 & 0.062 & 0.055 & 0.057 \\ 
  REF   & 0.586 & 0.134 & 0.084 & 0.063 & 0.054 & 0.057 \\ 
  \hline
 \hline
 \bf{M3} & $\sigma_1=0.05$ & $\sigma_2=0.1$ & $\sigma_3=0.15$ & $\sigma_ 4=0.2$ & $\sigma_5=0.4$ & $\sigma_6=0.8$ \\
  \hline
AsF   & 1.000 & 0.479 & 0.190 & 0.094 & 0.035 & 0.029 \\ 
  RPM   & 0.991 & 0.442 & 0.201 & 0.118 & 0.069 & 0.059 \\ 
  Fb   & 0.994 & 0.319 & 0.086 & 0.033 & 0.009 & 0.006 \\ 
  GPF   & 1.000 & 0.610 & 0.284 & 0.165 & 0.073 & 0.054 \\ 
  FP   & 1.000 & 0.615 & 0.275 & 0.152 & 0.065 & 0.057 \\ 
  IPT   & 1.000 & 0.586 & 0.260 & 0.150 & 0.069 & 0.050 \\ 
  F-max   & 1.000 & 0.628 & 0.230 & 0.141 & 0.068 & 0.050 \\ 
  p-min   & 1.000 & 0.527 & 0.168 & 0.096 & 0.041 & 0.036 \\ 
  GFAM   & 1.000 & 0.637 & 0.260 & 0.146 & 0.071 & 0.057 \\ 
  GFAC   & 1.000 & 0.659 & 0.267 & 0.160 & 0.063 & 0.059 \\ 
  REF   & 1.000 & 0.634 & 0.259 & 0.138 & 0.064 & 0.049 \\ 
 \hline
 \bf{M4} & $\sigma_1=0.05$ & $\sigma_2=0.1$ & $\sigma_3=0.15$ & $\sigma_ 4=0.2$ & $\sigma_5=0.4$ & $\sigma_6=0.8$ \\
  \hline
AsF   & 0.813 & 0.202 & 0.086 & 0.053 & 0.035 & 0.025 \\ 
  RPM   & 0.686 & 0.195 & 0.096 & 0.066 & 0.061 & 0.052 \\ 
  Fb   & 0.694 & 0.113 & 0.041 & 0.018 & 0.010 & 0.006 \\ 
  GPF   & 0.888 & 0.290 & 0.134 & 0.106 & 0.077 & 0.054 \\ 
  FP   & 0.876 & 0.280 & 0.131 & 0.096 & 0.065 & 0.052 \\ 
  IPT   & 0.861 & 0.265 & 0.125 & 0.093 & 0.069 & 0.048 \\ 
  F-max   & 0.548 & 0.156 & 0.098 & 0.069 & 0.068 & 0.043 \\ 
  p-min   & 0.468 & 0.116 & 0.061 & 0.060 & 0.042 & 0.029 \\ 
  GFAM   & 0.617 & 0.185 & 0.095 & 0.081 & 0.067 & 0.045 \\ 
  GFAC   & 0.623 & 0.181 & 0.097 & 0.072 & 0.067 & 0.056 \\ 
  REF   & 0.574 & 0.162 & 0.098 & 0.066 & 0.067 & 0.042 \\ 
  \hline
\end{tabular}
\end{table}

\begin{table}[ht]
\caption{The proportions of rejections (at level 0.05) over 1000 runs in the case of Brownian errors for models M1, M2, M3, M4. See text for the model specifications and {Table \ref{table:tests} for descriptions of different test abbreviations}.}
\label{Brown}
\centering
\begin{tabular}{lcccccc}
  \hline
 \bf{M1} & $\sigma_1=0.05$ & $\sigma_2=0.1$ & $\sigma_3=0.15$ & $\sigma_ 4=0.2$ & $\sigma_5=0.4$ & $\sigma_6=0.8$ \\
  \hline
AsF   & 0.066 & 0.063 & 0.060 & 0.060 & 0.060 & 0.069  \\ 
  RPM   & 0.055 & 0.039 & 0.042 & 0.042 & 0.038 & 0.045  \\ 
  Fb   & 0.033 & 0.032 & 0.038 & 0.032 & 0.030 & 0.039  \\ 
  GPF   & 0.068 & 0.071 & 0.068 & 0.071 & 0.071 & 0.077  \\ 
  FP   & 0.047 & 0.056 & 0.050 & 0.048 & 0.048 & 0.054  \\ 
  IPT   & 0.050 & 0.044 & 0.053 & 0.056 & 0.049 & 0.062  \\ 
  F-max   & 0.058 & 0.041 & 0.059 & 0.050 & 0.047 & 0.058 \\ 
  p-min   & 0.046 & 0.042 & 0.050 & 0.047 & 0.045 & 0.057  \\ 
  GFAM   & 0.047 & 0.047 & 0.056 & 0.050 & 0.052 & 0.054  \\ 
  GFAC   & 0.053 & 0.050 & 0.066 & 0.044 & 0.049 & 0.053 \\ 
  REF   & 0.053 & 0.037 & 0.054 & 0.048 & 0.048 & 0.057 \\ 
   \hline
 \hline
 \bf{M2} & $\sigma_1=0.05$ & $\sigma_2=0.1$ & $\sigma_3=0.15$ & $\sigma_ 4=0.2$ & $\sigma_5=0.4$ & $\sigma_6=0.8$ \\
  \hline
AsF   & 0.660 & 0.134 & 0.091 & 0.080 & 0.077 & 0.062  \\ 
  RPM   & 0.993 & 0.623 & 0.255 & 0.134 & 0.065 & 0.049  \\ 
  Fb   & 0.393 & 0.065 & 0.048 & 0.037 & 0.042 & 0.026  \\ 
  GPF   & 1.000 & 0.663 & 0.279 & 0.171 & 0.102 & 0.071 \\ 
  FP   & 0.645 & 0.106 & 0.066 & 0.064 & 0.065 & 0.049  \\ 
  IPT   & 1.000 & 0.600 & 0.227 & 0.129 & 0.074 & 0.049 \\ 
  F-max   & 1.000 & 0.958 & 0.599 & 0.337 & 0.110 & 0.059 \\ 
  p-min   & 1.000 & 0.954 & 0.584 & 0.330 & 0.105 & 0.050 \\ 
  GFAM   & 1.000 & 0.949 & 0.548 & 0.321 & 0.112 & 0.053 \\ 
  GFAC   & 1.000 & 0.930 & 0.540 & 0.308 & 0.116 & 0.048 \\ 
  REF   & 1.000 & 0.955 & 0.598 & 0.338 & 0.108 & 0.053 \\ 
  \hline
 \hline
 \bf{M3} & $\sigma_1=0.05$ & $\sigma_2=0.1$ & $\sigma_3=0.15$ & $\sigma_ 4=0.2$ & $\sigma_5=0.4$ & $\sigma_6=0.8$ \\
  \hline
AsF   & 1.000 & 0.506 & 0.194 & 0.122 & 0.073 & 0.055 \\ 
  RPM   & 1.000 & 0.997 & 0.894 & 0.652 & 0.159 & 0.064 \\ 
  Fb   & 0.997 & 0.230 & 0.098 & 0.047 & 0.040 & 0.031 \\ 
  GPF   & 1.000 & 1.000 & 1.000 & 0.996 & 0.329 & 0.099 \\ 
  FP   & 1.000 & 0.455 & 0.158 & 0.092 & 0.057 & 0.042 \\ 
  IPT   & 1.000 & 1.000 & 1.000 & 0.994 & 0.259 & 0.073 \\ 
  F-max   & 1.000 & 1.000 & 1.000 & 1.000 & 1.000 & 0.608  \\ 
  p-min   & 1.000 & 1.000 & 1.000 & 1.000 & 1.000 & 0.592 \\ 
  GFAM   & 1.000 & 1.000 & 1.000 & 1.000 & 0.996 & 0.436 \\ 
  GFAC   & 1.000 & 1.000 & 1.000 & 1.000 & 0.994 & 0.476  \\ 
  REF   & 1.000 & 1.000 & 1.000 & 1.000 & 1.000 & 0.598  \\ 
  \hline
 \hline
 \bf{M4}& $\sigma_1=0.05$ & $\sigma_2=0.1$ & $\sigma_3=0.15$ & $\sigma_ 4=0.2$ & $\sigma_5=0.4$ & $\sigma_6=0.8$ \\
  \hline
AsF   & 0.746 & 0.226 & 0.132 & 0.096 & 0.105 & 0.063  \\ 
  RPM   & 0.920 & 0.288 & 0.110 & 0.072 & 0.060 & 0.040 \\ 
  Fb   & 0.590 & 0.144 & 0.077 & 0.060 & 0.053 & 0.038 \\ 
  GPF   & 1.000 & 0.667 & 0.295 & 0.190 & 0.115 & 0.074 \\ 
  FP   & 0.691 & 0.193 & 0.116 & 0.077 & 0.079 & 0.049 \\ 
  IPT   & 1.000 & 0.596 & 0.240 & 0.143 & 0.091 & 0.052 \\ 
  F-max   & 1.000 & 1.000 & 0.981 & 0.808 & 0.188 & 0.075 \\ 
  p-min   & 1.000 & 1.000 & 0.978 & 0.791 & 0.180 & 0.068  \\ 
  GFAM   & 1.000 & 0.999 & 0.903 & 0.628 & 0.155 & 0.070 \\ 
  GFAC   & 1.000 & 0.997 & 0.893 & 0.648 & 0.146 & 0.069 \\ 
  REF   & 1.000 & 1.000 & 0.981 & 0.798 & 0.187 & 0.074 \\ 
   \hline
   \end{tabular}
\end{table}

In the second part of our simulation study, we studied the robustness of all studied methods to heteroscedasticity. We computed the empirical significance levels of all tests in the case of Gaussian process error structure and a) with standard deviations $\sigma\cdot 0.5^{(i-1)}$, where $i$ is the group indicator, and
b) with the same standard deviations $\sigma\cdot 0.5^{(i-1)}$ and further with scale parameter of the Gaussian error process equal to $i/10$.
The setting a) corresponds to the case of unequal variances, where our tests are asymptotically exact. In the setting b), the variances as well as covariance structures are unequal.
First, we used the tests without any corrections (Table \ref{unequal} rows 1-9). Second, we explored the corrected versions of our three tests (Table \ref{unequal} rows 12-14) as well as of IPT and RPM for comparison. For IPT, a variance transformation similar to the transformation \eqref{Sij2} employed for GFAM and GFAC tests was used. The RPM method relies instead on the variance correction of the $F$-statistic, similarly as REF method. 
The other methods of the table were not corrected for heteroscedasticity due to the fact that their implementation in R did not support it, but it is shown in \citet{CuevasEtal2004} that the AsF is robust to heteroscedasticity.


\begin{table}[ht]
\caption{The proportions of rejections (at level 0.05) over 1000 runs for model M1 for the Gaussian error process and heteroscedastic case. The results shown in last three rows correspond to our three tests with correction for unequal variances. {Also the RPM and IPT methods (below the horizontal line) were corrected, while the other tests (above the horizontal line) were used without corrections.}}
\label{unequal}
\centering
\begin{tabular}{c|c}
\hline
{\bf Unequal variances} & {\bf Unequal variances and covariances}\\
\begin{tabular}{lccc}
  \hline
  & $\sigma_1=0.05$ & $\sigma_2=0.1$ & $\sigma_3=0.15$ \\
  \hline
AsF   & 0.031 & 0.035 & 0.017\\ 
  Fb   & 0.007 & 0.008 & 0.006\\ 
  GPF   & 0.103 & 0.100 & 0.084\\ 
  FP   & 0.087 & 0.095 & 0.072\\ 
  F-max   & 0.282 & 0.267 & 0.272\\ 
  p-min   & 0.181 & 0.184 & 0.180\\ 
  GFAM   & 0.173 & 0.163 & 0.178\\ 
  GFAC   & 0.204 & 0.194 & 0.184\\ 
  REF   & 0.241 & 0.226 & 0.234\\ \hline
  RPM   & 0.107 & 0.103 & 0.117\\ 
 IPT   & 0.109 & 0.116 & 0.103\\ 
   GFAMU   & 0.111 & 0.107 & 0.081\\ 
  GFACU   & 0.074 & 0.063 & 0.054\\ 
  REFU   & 0.156 & 0.168 & 0.163\\ 
   \hline
   \end{tabular}
   &
   \begin{tabular}{lccc}
  \hline
  & $\sigma_1=0.05$ & $\sigma_2=0.1$ & $\sigma_3=0.15$ \\
  \hline
AsF   & 0.032 & 0.029 & 0.032\\ 
  Fb   & 0.011 & 0.006 & 0.011\\ 
  GPF   & 0.090 & 0.081 & 0.076\\ 
  FP   & 0.077 & 0.073 & 0.066\\ 
  F-max   & 0.285 & 0.279 & 0.272\\ 
  p-min   & 0.208 & 0.202 & 0.195\\ 
  GFAM   & 0.154 & 0.158 & 0.147\\ 
  GFAC   & 0.198 & 0.186 & 0.160\\ 
  REF   & 0.244 & 0.247 & 0.228\\ \hline 
  RPM   & 0.096 & 0.096 & 0.101\\ 
  IPT   & 0.104 & 0.102 & 0.093\\ 
   GFAMU   & 0.067 & 0.071 & 0.079\\ 
  GFACU   & 0.042 & 0.048 & 0.030\\ 
  REFU   & 0.144 & 0.143 & 0.124\\ 
   \hline
   \end{tabular}
   
   \end{tabular}
   
\end{table}

The second part of the study shows that the methods based on the maximum ($F$-max, $p$-min, GFAM, GFAC, REF) are much more sensitive to the heteroscedasticity than the methods based on the integral principle (Fb, GPF, IPT). The FP method, which is based on the basis representation, was least affected by heteroscedasticity. The AsF and Fb methods were clearly conservative as in the homoscedastic case. 
Considering the methods with correction for unequal variances, i.e.\ GFAMU, GFACU, REFU, IPT, RPM (Table \ref{unequal} rows 10-14), the GFACU test was the least liberal method in the small sample case of ten functions per group. The GFACU test was even less liberal than FP. The unequality of covariance structures did not affect the liberality of the methods. Thus, we conclude that our variance correction by transformation \eqref{Sij2} of functions can be used even for small sample sizes with expecting small liberality.

In the third part of our simulation study we took the model M3 and extended it for ten groups, considering
$$\text{M: } T_{ij}(r) = r^{i/5} (1-r)^{6-i/5} + e_{ij}(r), \; i=2, \ldots , 11, j=1, \ldots , 10.$$
We used again the two correlation structures in the model and six levels of contamination as in the first part of the study. Table \ref{10groups} summarizes the results both for the Gaussian process error (upper part) and the Brownian error case (lower part).
The relations between powers of different methods were the same as in the case of the three groups. Also there was no observable decrease in the power for the GFAM and GFAC methods with respect to the REF and $p$-min methods in the Gaussian process error case. On the other hand, there was such decrease in the Brownian error case. 
This loss of power in the GFAM and GFAC tests can be prevented by increasing number of permutations: We performed the experiment also with 10000 permutations and obtained very similar powers as with 2000 permutations for all the other methods except GFAM and GFAC: the power of GFAM increased from 0.360 to 0.884 and the power of GFAC from 0.360 to 0.688 in the case of $\sigma_6$ and Brownian errors. Thus the power of the graphical tests was comparable to the power of the REF and $p$-min tests with 10000 permutations.

\begin{table}[ht]
\caption{The proportions of rejections (at level 0.05) over 1000 runs for model M.  The Gaussian error process cases are shown in the upper part and the Brownian error cases are shown in the lower part of the table. See text for the model specification and {Table \ref{table:tests} for descriptions of different test abbreviations}.}
\label{10groups}
\centering
\begin{tabular}{lcccccc}
  \hline
 \bf{iid} & $\sigma_1=0.05$ & $\sigma_2=0.1$ & $\sigma_3=0.15$ & $\sigma_ 4=0.2$ & $\sigma_5=0.4$ & $\sigma_6=0.8$ \\
  \hline
AsF   & 1.000 & 0.950 & 0.493 & 0.262 & 0.057 & 0.045 \\ 
  RPM   & 1.000 & 0.982 & 0.573 & 0.275 & 0.106 & 0.055 \\ 
  Fb   & 1.000 & 0.899 & 0.271 & 0.065 & 0.007 & 0.005 \\ 
  GPF   & 1.000 & 0.996 & 0.743 & 0.421 & 0.085 & 0.092 \\ 
  FP   & 1.000 & 0.996 & 0.725 & 0.414 & 0.077 & 0.071 \\ 
  IPT   & 1.000 & 0.990 & 0.715 & 0.394 & 0.102 & 0.107 \\ 
  F-max   & 1.000 & 0.997 & 0.775 & 0.408 & 0.112 & 0.065 \\ 
  p-min   & 1.000 & 0.997 & 0.704 & 0.344 & 0.095 & 0.049 \\ 
  GFAM   & 1.000 & 1.000 & 0.786 & 0.449 & 0.120 & 0.077 \\ 
  GFAC   & 1.000 & 1.000 & 0.842 & 0.521 & 0.119 & 0.078 \\ 
  REF   & 1.000 & 1.000 & 0.760 & 0.425 & 0.106 & 0.078 \\ 
   \hline
 \hline
 \bf{Brown} & $\sigma_1=0.05$ & $\sigma_2=0.1$ & $\sigma_3=0.15$ & $\sigma_ 4=0.2$ & $\sigma_5=0.4$ & $\sigma_6=0.8$ \\
  \hline
AsF   & 1.000 & 0.996 & 0.656 & 0.364 & 0.092 & 0.064 \\ 
  RPM   & 1.000 & 1.000 & 1.000 & 0.996 & 0.544 & 0.104 \\ 
  Fb   & 1.000 & 0.944 & 0.208 & 0.148 & 0.036 & 0.012 \\ 
  GPF   & 1.000 & 1.000 & 1.000 & 1.000 & 0.724 & 0.148 \\ 
  FP   & 1.000 & 1.000 & 0.564 & 0.268 & 0.048 & 0.052 \\ 
  IPT   & 1.000 & 1.000 & 1.000 & 1.000 & 0.656 & 0.124 \\ 
  F-max   & 1.000 & 1.000 & 1.000 & 1.000 & 1.000 & 0.912 \\ 
  p-min   & 1.000 & 1.000 & 1.000 & 1.000 & 1.000 & 0.892 \\ 
  GFAM   & 1.000 & 1.000 & 1.000 & 1.000 & 0.976 & 0.360 \\ 
  GFAC   & 1.000 & 1.000 & 1.000 & 1.000 & 1.000 & 0.360  \\ 
  REF   & 1.000 & 1.000 & 1.000 & 1.000 & 1.000 & 0.904 \\ 
   \hline
   \end{tabular}
\end{table}

In the last part of the simulation study we studied the powers of all test with respect to increasing discretization of the functions. For this purpose we simulated the model M3 with Gaussian process error with standard deviation equal to 0.15. Table \ref{discrete} shows the estimated powers for five different levels of discretizations expressed by the number of observations of functions.
The powers did not decrease with increasing level of discretization. The only exception was the AsF method, whose power decreased, which is correspondence with our finding in the first part of the study that this method was conservative for resolution with 100 points.

\begin{table}[ht]
\caption{The proportions of rejections (at level 0.05) over 1000 runs for model M3, Gaussian process error with standard deviation equal to 0.15 and various levels of discretizations {({\bf $K$} = number of discretized values of functions)}. }
\label{discrete}
\centering
\begin{tabular}{crrrrr}
  \hline
{\bf $K$} & 25 & 50 & 100 & 200 & 400 \\ 
  \hline
AsF   & 0.157 &	0.16	&0.164&	0.160&	0.165 \\ 
  RPM   & 0.183	&0.192	&0.178&	0.19&	0.182 \\ 
  Fb   & 0.07	&0.075&	0.077&	0.075&	0.078 \\ 
  GPF   & 0.254&	0.264&	0.273&	0.276&	0.274 \\ 
  FP   & 0.234&	0.248&	0.254&	0.255&	0.250 \\ 
  IPT   & 0.235&	0.248&	0.252&	0.253&	0.253 \\ 
  F-max   & 0.228&	0.252&	0.255&	0.246&	0.236 \\ 
  p-min   &0.210&	0.196&	0.190&	0.194&	0.033\\ 
  GFAM   & 0.234&	0.254&	0.268&	0.267&	0.256 \\ 
  GFAC   & 0.229&	0.254&	0.268&	0.254&	0.264 \\ 
  REF   & 0.233	&0.245&	0.249&	0.238&	0.231 \\ 
   \hline
\end{tabular}
\end{table}

\section{Fiscal decentralization example}\label{sec:DS}
The topic of fiscal federalism was brought about into the normative theory of public finance in the middle of twentieth century. The main issue was to solve the extent to which fiscal competences and responsibilities should be decentralized from central to sub-central levels of government. The gradual development of the theory of fiscal decentralization led to distinguishing between the first and the second generation theories of fiscal decentralization, as explained in details by \citet{Oates2005} and \citet{Vo2010}.

Generally, there are two types of empirical studies on fiscal decentralization. Within the first type, the concern is in the consequences of fiscal decentralization in terms of economic growth and the growth of public sector. The second type of studies deal with the determinants of fiscal decentralization including a growing body of literature dealing with the issues of globalization, economic and political integration and its consequences on decentralization or secession. {Our empirical example deals with the effects of European integration on fiscal decentralization dynamics in individual European countries. }

\subsection{ Decentralization under European economic and political integration}\label{se:Decentr}
The decentralization has been a characteristic feature of social development in many democratic countries since the last decades of 20th century. The usual presumption is that the federated countries are more decentralized than unitary ones. The process of decentralization, however, is not derived only from the switch to the federal structure, but it is usually more gradual and it is influenced by a number of factors. Even though the institutional (constitutional) changes might be crucial, there are usually many gradual and subtle changes in de facto decentralization. Moreover, as \citet{Arzaghi2005} noted, constitutional changes are discrete events which in certain contexts may be difficult culturally and politically to achieve. The gradual changes are more likely to be reflected in a continuous measure, such as the ratio of state and local governments in total general government expenditures or revenue. 

Recently, the issue of centralization versus decentralization of government has attracted attention in Europe. 
On the one hand, there are efforts to further integrate or even federalize the European Union (EU), on the other hand there is resistance to further integration, or even process of (br)exit from the EU (the case of United Kingdom). Also some secession tendencies at the sub-national level of individual countries, i.e.\ in Spain, Italy, Belgium or the UK, are quite strong \citep{Sedova2017}. The prominent argument behind these tendencies is an insufficient (fiscal) decentralization.
The EU and its member countries are experiencing two parallel tendencies of decentralization and centralization. 

\citet{Ackerman1997} claims that there is a continuum between international treaties based integration and the federal constitutions and thus between international organizations and federations. European Community, even before it became the European Union, has moved a long way along the continuum towards the federation. This is true mostly with regards to the capacity to legislate or regulate. The EU budget is still rather small, amounting to approximately one percent of gross national income (GNI). According to \citet{Spahn2015}, the EU is entity sui generis, not a federation, although this might be contested. The EU already has institutions of a federation including a Second Chamber in the form of the European Council and powerful exclusive policy competencies in competition and commerce. 

A deep insight into the causes of decentralization or even secession trends is provided by \citet{Alesina1997}. They formulated the trade-off between the benefits of large jurisdictions and the cost of heterogeneity of large and diverse populations, which determine the size of countries. By reducing the political and economic transaction cost, economic integration extends the size of market and lowers the benefits of large jurisdictions, thus enhancing incentives to secession. They conclude that the democratization (as is in the case of the EU) lead to secessions. Decentralization can however be associated with lower cost and may be preferred over secession \citep{Bolton1997}. Because the EU enlargements, whilst keeping the European diversity, have gradually created large single internal market, the EU economy as a whole became much less open compared to individual countries. Therefore, we can hypothesize tendency to the European integration conditioned decentralization of government budgets in individual EU countries.
Similar conclusions of positive effects of economic, political and social integration on fiscal decentralization were found by \citet{Stegarescu2009} and \citet{Ermini2014} for OECD countries. National governments in the integrating EU therefore run the risk of getting squeezed between the supranational and the sub-national levels of government. Lower government expenditure centralization ratio may be expected in this case.


\citet{Rodrik1998} formulated the positive effect of economic openness on the size of the public sector as well as on the fiscal centralization due to higher government expenditures for centralized redistribution and macroeconomic stabilization. If the economic integration strengthens incentives to fiscally decentralize, under some circumstances, an increased integration may cause fiscal centralization. Such circumstances may come with macroeconomic imbalances in economic crisis.

\subsection{Decentralization characteristic, data and hypothesis formulation}
To analyze the fiscal decentralization, a suitable characteristic is needed.  There are variety of approaches of expressing the fiscal decentralization phenomenon \citep[for complex overview see][]{Stegarescu2005,Vo2008}. This paper follows in principle the approach of \citet{Cerniglia2003} and \citet{Arzaghi2005}, using  the ratio of centralization. The advantage of this simple approach is twofold: it avoids the problems with various, complicated and not easily comparable structures of decentralized levels of governments, and it gives the largest dataset.

We use the government expenditure centralization ratio (GEC) in percent. It is the ratio of central government expenditure to the total general government expenditure. Because it includes all kinds of government expenditure (consumption, investment and transfers), it is the most general measure of expenditure decentralization. Data were collected from the \citet{Eurostat} database. European countries were selected in order to achieve the maximum size of dataset. Only those countries were included, where the data were available from 1995 to 2016 without interruption. Finally, 29 countries were classified into three groups in the following way:

\begin{enumerate}
    \item[Group 1:] Countries joining EC between 1958 and 1986 (Belgium, Denmark, France, Germany (until 1990 former territory of the FRG), Greece, Ireland, Italy, Luxembourg, Netherlands, Portugal, Spain, United Kingdom. These countries have long history of European integration, representing the core of integration process.

\item[Group 2:] Countries joining the EU in 1995 (Austria, Sweden, Finland) and 2004 (Malta, Cyprus), except CEEC (separate group), plus highly economically integrated non-EU countries, EFTA members (Norway, Switzerland). Countries in this group have been, or in some case even still are standing apart from the integration mainstream. Their level of economic integration is however very high.

\item[Group 3:] Central and Eastern European Countries (CEEC), having similar features in political end economic history. The process of economic and political integration have been initiated by political changes in 1990s. CEEC joined the EU in 2004 and 2007 (Bulgaria, Czech Republic, Estonia, Hungary, Latvia, Lithuania, Poland, Romania, Slovakia, Slovenia, data for Croatia joining in 2013 are incomplete, therefore not included).
\end{enumerate}
Finally, based on the discussion presented in Section \ref{se:Decentr}, we formulate the null hypothesis that the trend of the fiscal centralization is same in all three groups of countries. 

\subsection{Data analysis}

The data were first centred with respect to country average in order to remove the differences in absolute {values} of GEC between countries{ and in order to keep the shape of GEC functions}. Thus we study the functions
$CGEC_{ij}(r)=GEC_{ij}(r)-\frac{1}{22}\sum_{r=1995}^{2016}GEC_{ij}(r), j=\{1,2,3\}, i=1,\ldots , n_j$, where $n_1=12, n_2=7, n_3=10$ {and $r=1995,1996,\dots,2016$}. The curves are shown in Figure \ref{fig:GE_C_curves}.
\begin{figure}
    \centering
    \includegraphics[width=\textwidth]{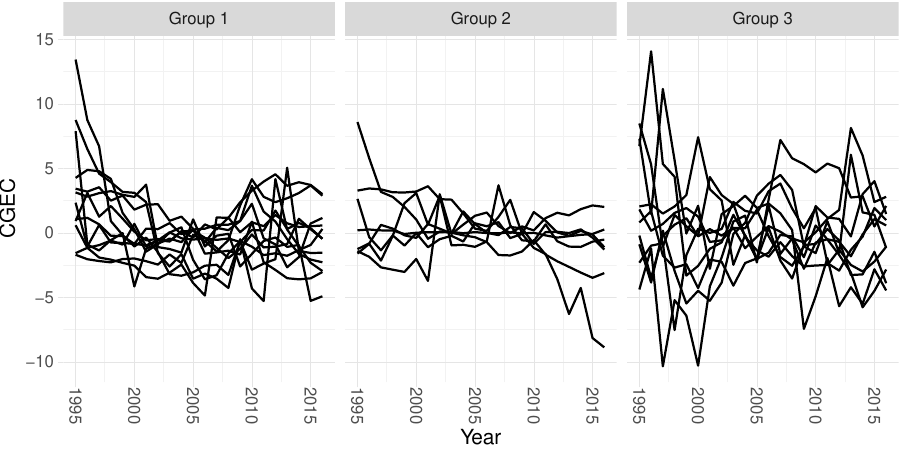}
    \caption{The centred GEC index for the years from 1995 to 2016 in the three groups.}
    \label{fig:GE_C_curves}
\end{figure}

First we checked the assumption of equality of covariance structure which is required by our tests. This we propose to do by testing the equality of lag 1 covariances. Figure \ref{fig:CovarianceTest} shows the result of the REF test applied to {the transformed functions \eqref{eq:covariance} ($p=0.392$)}. 
\begin{figure}
    \centering
   r \includegraphics[width=0.5\textwidth]{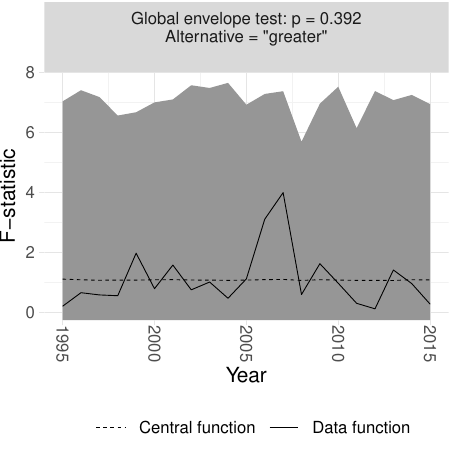}
    \caption{The test for equality of lag 1 covariances of the centred GEC indices in the three groups using the rank envelope $F$-type test (REF).}
    \label{fig:CovarianceTest}
\end{figure}

The next step is to decide if the correction for unequal variances should be used. Figure \ref{fig:VarianceTest} shows the results of the GFAM test applied {to the transformed functions \eqref{eq:variance}}. Since the global $p$-value is  0.163, we have no evidence that the group variances differ and we therefore prefer to use no correction. Figure \ref{fig:VarianceTest} also shows the mean absolute deviation across the three groups with the global $95\%$ envelope reflecting the overall variation of the mean absolute deviation among all groups {over the years}.

\begin{figure}
    \centering
    \includegraphics[width=\textwidth]{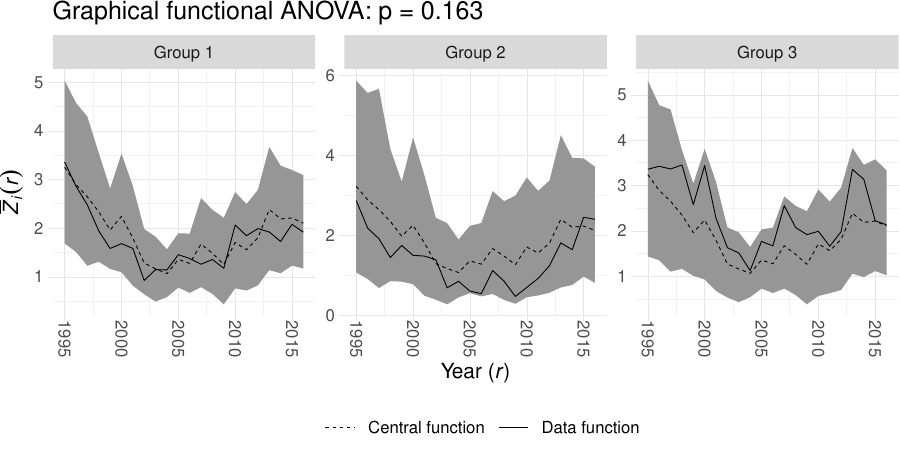}
    \caption{The GFAM test for equality of variances of the centred GEC indices in the three groups.}
    \label{fig:VarianceTest}
\end{figure}

Finally, we performed the GFAM and GFAC tests for equality of means. The Figure \ref{fig:GE_C_GFAM} shows the mean centred GEC functions across the groups and their time developments together with the global $95\%$ envelope reflecting the overall variation of the mean function among all groups. 
Figure \ref{fig:GE_C_GFAC} shows the group differences of the mean centred GEC functions and their time developments together with the global $95\%$ envelope reflecting the variation of these differences. Both tests shows the deviation from the null hypothesis in the year 2006 and the difference of the groups is supported by $p$-values equal to 0.045 and 0.021 respectively. According to the GFAM test, the first group is significantly different from the other two. According to the GFAC test in the post-hoc fashion, only groups 1 and 3 significantly differ. All the shown test were performed with 9999 permutations.

For comparison, we also performed the other tests of the simulation study (see Table \ref{table:tests}).
First, our REF test was borderline significant ($p=0.055$). The data statistic had a peak very close to envelope for the year 2006, but the test was not able to reject the null hypothesis at the strict significance level 0.05 (figure omitted). 
The $p$-values of the other tests were for AsF equal to 0.467, for  
RPM method equal to 0.375, Fb 0.531, GPF 0.304, FP 0.559, IPT 0.265 and $F$-max 0.046.
Since the deviation is present mainly only in one year, the integral based methods are not capable to detect this small difference, whereas the maximum type method $F$-max and all our methods are able to detect this difference.

\begin{figure}
    \centering
    \includegraphics[width=\textwidth]{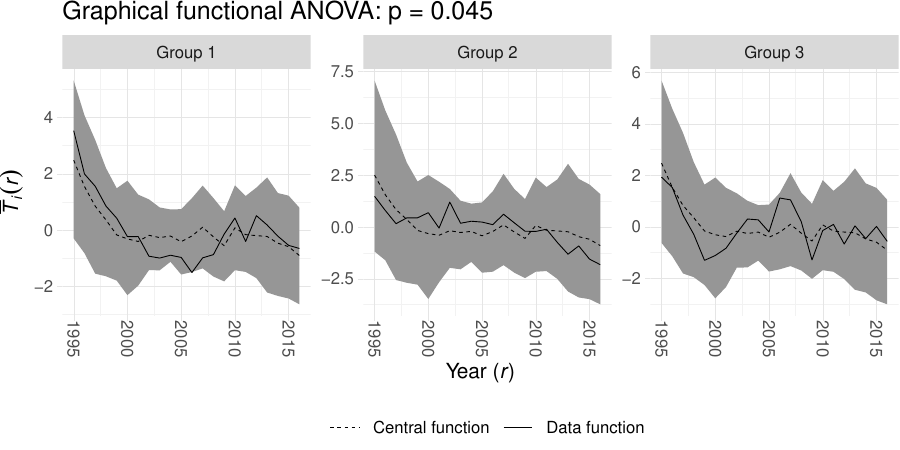}
    \caption{ The GFAM test for equality of means of the centred GEC indices in the three groups using the group means.}
    \label{fig:GE_C_GFAM}
\end{figure}

\begin{figure}
    \centering
    \includegraphics[width=\textwidth]{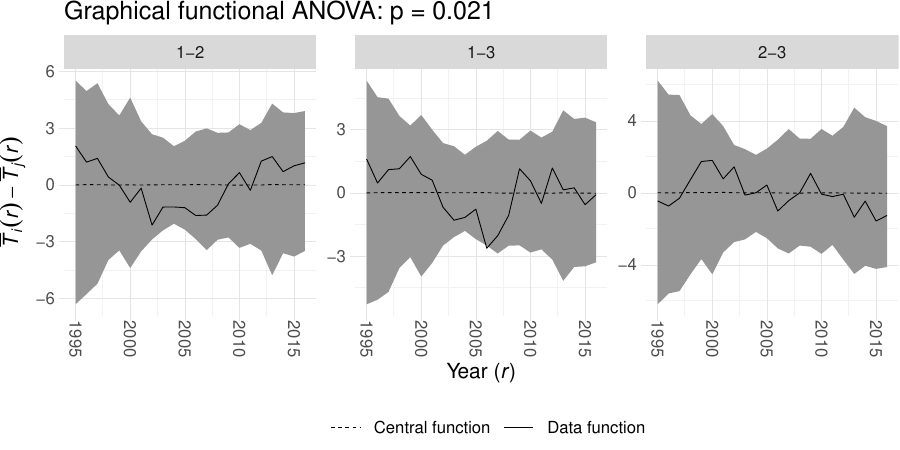}
    \caption{The GFAC test for equality of means of the centred GEC indices in the three groups using the group mean contrasts (comparison of groups 1--2, 1--3 and 2--3).}
    \label{fig:GE_C_GFAC}
\end{figure}

\subsection{Interpretation of the tests results}

{The permutation based tests $F$-max, GFAM and GFAC found differences between the groups. The GFAM test (} Figure \ref{fig:GE_C_GFAM}) shows greater {downward departure from the central function of government expenditure centralization in the period between 2001 and 2008} in the group 1 than in other groups. This tendency is supported by significant difference observed in the year 2006. This period roughly corresponds to the period between Euro introduction and start of financial crises. The tendency of the countries in the first group to centralize after the year 2008 could be caused by the need of central governments to increase the control over the general government budgets at the onset of the crisis. Also the well known phenomenon of rather pro-cyclical fiscal behaviour of sub-national governments supports this argument.

{The GFAC test (}Figure \ref{fig:GE_C_GFAC}) shows that the difference in trends of mean group government expenditure centralization is realized between groups 1 and 3. This result is consistent with the expectation, because the difference in the depth of integration between these groups is the largest.


\section{Discussion}\label{sec:discussion}

A new one-way graphical functional ANOVA test was introduced in this paper. Under the assumption that the all the data come from the same distribution, the test has exact type I error. It provides a graphical interpretation which is essential for the interpretation of the results. {Also tests and corrections for heteroscedasticity were presented.} Two different test vectors can be used leading to two different graphical interpretations. The first option compares every group with the rest of the groups. The second option compares differences between every pair of groups similarly like a Tukey post-hoc test in the univariate ANOVA. A positive side effect of our method is that the post-hoc test is provided together with the main ANOVA procedure at the given significance level. 

Since the proposed test works in a highly dimensional multivariate settings, no smoothness of the functions is required. On the other hand, the same discretization of functions is required for every function. If this is not the case, a smoothing technique has to be applied followed by further identical discretization of functions.

Our new methods were compared to the other functional ANOVA procedures, available through the software R, with respect to their power.  The new graphical {tests} had comparable power with respect to other procedures {in our simulation study}. {It was shown that the presented correction for unequal variances, which transforms the functions, can be used even with low sample sizes. Its robustness with respect to unequal covariance structure was also shown.}
 {In addition,} the proposed methods did not loose their power when the functions were more densely discretized. 


Importantly, our simulation study shows that there is no procedure which would be uniformly more powerful than our proposed {tests}. Therefore, we believe that our {tests} are useful in practice due to their graphical interpretation and post-hoc nature. 
{As shown by our simulation study,} our methods can loose some power when {the number of groups to compare is large}. In such a case where the test vector is very long, the number of permutations has to be increased in order to eliminate this problem.

Our new {tests} were designed for the one-way functional ANOVA design. A question of our future research is how these procedures can be extended into multi-way design. Since the permutation of the functional residuals leads to a liberal method, the problem has to be solved in a more complex way.

Our methods show graphically the region of rejections in the family wise error rate. The recent method of \citet{PiniEtAl2018} can also show the regions of rejections but in the sense of interval wise control of the error rate. We plan to compare these two approaches in the future.

\section*{Acknowledgements}
Mrkvi\v cka has been financially supported by the Grant Agency of Czech Republic (Project Nos.\ 16-03708S and 19-04412S) and Mari Myllym{\"a}ki by the Academy of Finland (project numbers 295100, 306875). Hahn's research has been supported by  the Centre for Stochastic Geometry and Advanced Bioimaging, funded by the Villum Foundation. 

\bibliographystyle{asa}



\end{document}